\documentclass{amsart}
\usepackage{amssymb}
\usepackage{amsmath}
\usepackage{amsthm}
\usepackage{mathrsfs}
\usepackage{colortbl}
\usepackage{graphicx}
\newtheorem{theorem}{Theorem}
\newtheorem{lemma}[theorem]{Lemma}

\newtheorem{remark}[theorem]{Remark}

\theoremstyle{definition}
\newtheorem{definition}[theorem]{Definition}

\newtheorem{proposition}[theorem]{Proposition}

\makeatletter

\@addtoreset{theorem}{section}
\makeatother

\makeatletter

\@addtoreset{equation}{section}
\makeatother

\begin{document}                                                 
\title[Multiphase Flow System with Surface Tension and Flow]{Thermodynamical Modeling of Multiphase Flow System with Surface Tension and Flow}                                 
\author[Hajime Koba]{Hajime Koba}                                
\address{Graduate School of Engineering Science, Osaka University,\\
1-3 Machikaneyamacho, Toyonaka, Osaka, 560-8531, Japan}                                  
\email{iti@sigmath.es.osaka-u.ac.jp}

\thanks{This work was partly supported by the Japan Society for the Promotion of Science (JSPS) KAKENHI Grant Number JP21K03326.}                                     
\keywords{Multiphase flow, Surface tension, Surface flow, Mathematical modeling, First law of thermodynamics, Energetic variational approach}                            
\subjclass[]{80M30, 35Q79, 76-10, 80-10, 35A15}                                
\begin{abstract}
We consider the governing equations for the motion of the viscous fluids in two moving domains and an evolving surface from both energetic and thermodynamic points of view. We make mathematical models for multiphase flow with surface flow by our energetic variational and thermodynamic approaches. More precisely, we apply our energy densities, the first law of thermodynamics, and the law of conservation of total energy to derive our multiphase flow system with surface tension and flow. We study the conservative forms and conservation laws of our system by using the surface transport theorem and integration by parts. Moreover, we investigate the enthalpy, the entropy, the Helmholtz free energy, and the Gibbs free energy of our model by applying the thermodynamic identity. The key idea of deriving surface tension and viscosities is to make use of both the first law of thermodynamics and our energy densities.
\end{abstract}       
\maketitle

\section{Introduction}\label{sect1}

\begin{figure}[htbp]
\includegraphics[width=12cm]{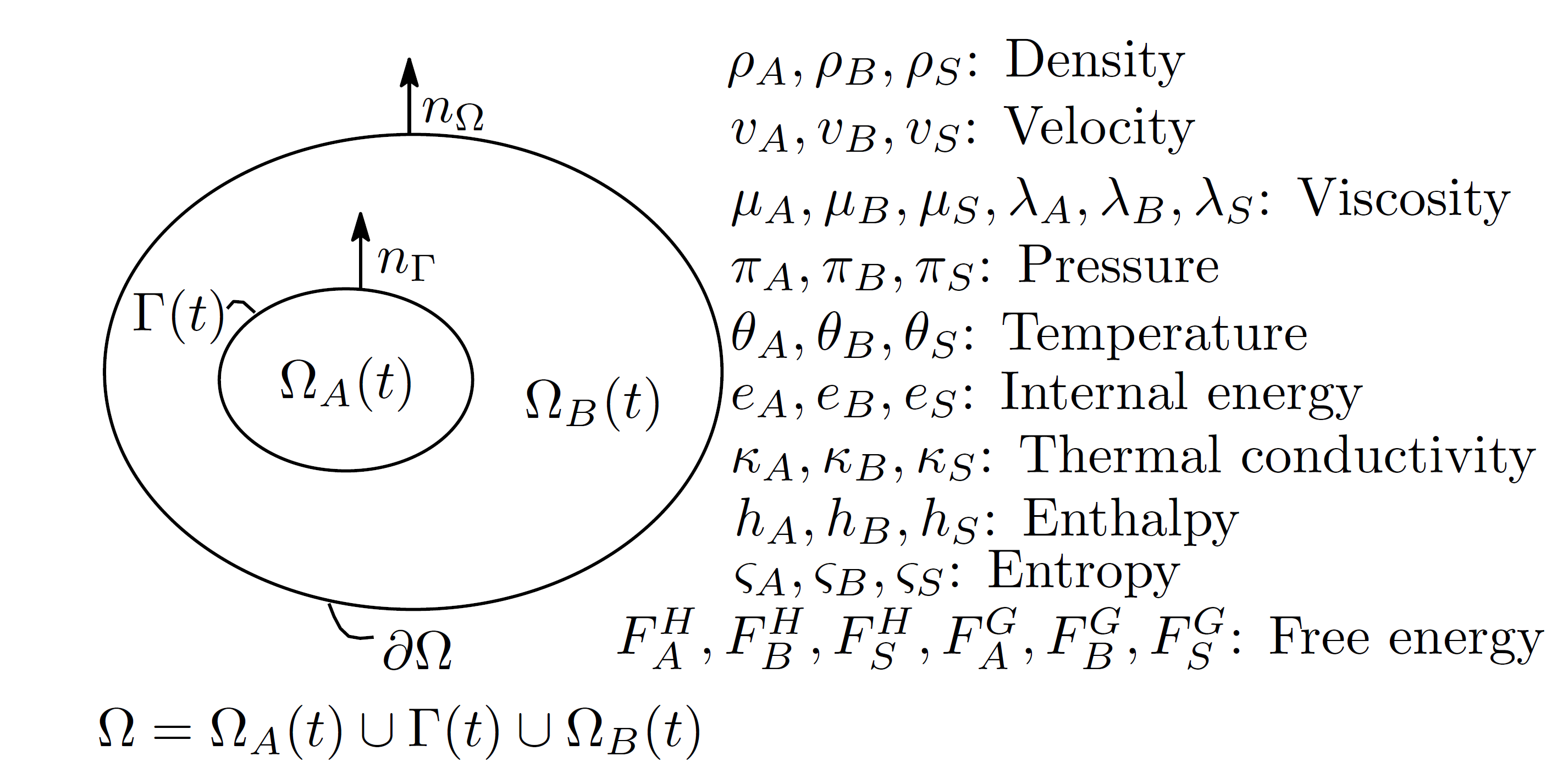}
\caption{Moving Domains, Surfaces and Notations}
\label{Fig1}
\end{figure}

We are interested in a mathematical modeling of a soap bubble floating in the air. When we focus on a soap bubble, we can see the fluid flow in the bubble. We call the fluid flow in the bubble a \emph{surface flow}. We can consider a surface flow as a fluid-flow on an evolving surface. To make a mathematical model for a soap bubble floating in the air, we have to study the dependencies among fluid-flows in two moving domains and surface flow. We consider the governing equations for the motion of the viscous fluids in the two moving domains and surface from both energetic and thermodynamic points of view. More precisely, we apply the first law of thermodynamics and our energy densities to derive our multiphase flow system with surface tension and flow.

Let us first introduce fundamental notations. Let $t \geq 0$ be the time variable, and $x( = { }^t (x_1 , x_2, x_3 ) ) \in \mathbb{R}^3$ the spatial variable. Fix $T >0$. Let $\Omega \subset \mathbb{R}^3$ be a bounded domain with a smooth boundary $\partial \Omega$. The symbol $n_\Omega = n_\Omega (x) = { }^t (n^\Omega_1 , n^\Omega_2 , n^\Omega_3 )$ denotes the unit outer normal vector at $x \in \partial \Omega$. Let $\Omega_A (t) (= \{ \Omega_A (t) \}_{0 \leq t < T} )$ be a bounded domain in $\mathbb{R}^3$ with a moving boundary $\Gamma (t)$. Assume that $\Gamma (t) (= \{ \Gamma (t) \}_{0 \leq t < T})$ is a smoothly evolving surface and is a closed Riemannian 2-dimensional manifold. The symbol $n_\Gamma = n_\Gamma ( x , t ) = { }^t (n^\Gamma_1 , n^\Gamma_2 , n^\Gamma_3)$ denotes the unit outer normal vector at $x \in \Gamma (t)$. For each $t \in [0,T)$, assume that $\Omega_A (t) \Subset \Omega$. Set $\Omega_B (t) = \Omega \setminus \overline{\Omega_A (t)}$. It is clear that $\Omega = \Omega_A (t) \cup \Gamma (t) \cup \Omega_B (t)$ (see Figure \ref{Fig1}). Set
\begin{multline*}
\Omega_{A,T} = \bigcup_{0< t < T} \{ \Omega_A (t) \times \{ t \} \},{ \ } \Omega_{B,T} = \bigcup_{0< t < T} \{ \Omega_B (t) \times \{ t \} \},\\
\Gamma_T = \bigcup_{0< t < T} \{ \Gamma (t) \times \{ t \} \},{ \ }\Omega_T = \Omega \times (0,T),{ \ }\partial \Omega_T = \partial \Omega \times (0,T).
\end{multline*}

In this paper we assume that the fluids in $\Omega_{A,T}$, $\Omega_{B,T}$, and $\Gamma_T$ are compressible ones. Let us state physical notations. For $\sharp = A , B,S$, let $\rho_\sharp = \rho_\sharp ( x , t)$, $v_\sharp = v_\sharp ( x , t) = { }^t (v^\sharp_1 , v^\sharp_2 , v^\sharp_3 )$, $\pi_\sharp = \pi_\sharp (x,t)$, $\theta_\sharp = \theta_\sharp ( x , t) $, $e_\sharp = e_\sharp (x,t)$, $\kappa_\sharp = \kappa_\sharp ( x , t)$ and $\mu_\sharp = \mu_\sharp (x,t)$, $\lambda_\sharp = \lambda_\sharp (x,t)$ be the \emph{density}, the \emph{velocity}, the \emph{pressure}, the \emph{temperature}, the \emph{internal energy}, the \emph{thermal conductivity}, and two \emph{viscosities} of the fluid in $\Omega_\sharp (t)$, where $\Omega_S (t) := \Gamma (t)$. The symbols $h_\sharp = h_\sharp (x,t)$, $\varsigma_\sharp = \varsigma_\sharp (x,t)$, $F^H_\sharp = F^H_\sharp (x,t)$, and $F^G_\sharp = F^G_\sharp (x,t)$ denote the \emph{enthalpy}, the \emph{entropy}, the \emph{Helmholtz free energy}, and the \emph{Gibbs free energy} of the fluid in $\Omega_\sharp (t)$, respectively (see Figure \ref{Fig1}). We call $\mu_\sharp$ the \emph{share viscosity} and $\mu_\sharp + \lambda_\sharp$ the \emph{dilatational viscosity}. In particular, we often call $\mu_S$ the \emph{surface share viscosity}, $\mu_S + \lambda_S$ the \emph{surface dilatational viscosity}. We assume that $\rho_\sharp$, $v_\sharp$, $\pi_\sharp$, $\theta_\sharp$, $e_\sharp$, $\kappa_\sharp$, $\mu_\sharp$, $\lambda_\sharp$, $h_\sharp$, $\varsigma_\sharp$, $F^H_\sharp$, and $F^G_\sharp$ are smooth functions in $\mathbb{R}^4$.

\begin{remark}\label{rem11} We call $v_S$ a \emph{total velocity}, and $\pi_S$ a \emph{total pressure}. Total velocity means that $v_S$ can be divided into surface velocity $u_S$ and motion velocity $w_S$, that is, $v_S = u_S + w_S$. Total pressure means one that includes surface pressure and tension. In this paper, we focus on the total velocity and the total pressure.
\end{remark}

Let us introduce several operators and notations. For each $f = f (x,t) \in C^1 (\mathbb{R}^4)$ and $V = V(x,t) = { }^t (V_1,V_2,V_3) \in [C^1 (\mathbb{R}^4)]^3$, $D_t^A f := \partial_t f + (v_A \cdot \nabla ) f $, $D_t^B f := \partial_t f + (v_B \cdot \nabla ) f $, $D_t^S f := \partial_t f + (v_S \cdot \nabla ) f $, ${\rm{grad}} f := \nabla f$, ${\rm{div}}V := \nabla \cdot V$, ${\rm{grad}}_\Gamma f := \nabla_\Gamma f$, ${\rm{div}}_\Gamma V := \nabla_\Gamma \cdot V$, $(V \cdot \nabla ) f := V_1 \partial_1 f + V_2 \partial_2 f + V_3 \partial_3 f$, $(V \cdot \nabla_\Gamma ) f := V_1 \partial^\Gamma_1 f + V_2 \partial^\Gamma_2 f + V_3 \partial^\Gamma_3 f$, where $\nabla := { }^t (\partial_1 , \partial_2 , \partial_3)$, $\partial_i := \partial/{\partial x_i}$, $\partial_t := \partial/{\partial t}$,  $\nabla_\Gamma := { }^t (\partial^\Gamma_1 , \partial^\Gamma_2 , \partial^\Gamma_3 ) $, and $\partial^\Gamma_i f := \sum_{j=1}^3(\delta_{ij} - n^\Gamma_i n^\Gamma_j ) \partial_j f = \partial_i f - n_i^\Gamma (n_\Gamma \cdot \nabla ) f$. Define the \emph{orthogonal projection $P_\Gamma$ to a tangent space} by
\begin{equation*}
P_\Gamma = P_\Gamma (x,t) = I_{3 \times 3} -n_\Gamma \otimes n_\Gamma = 
\begin{pmatrix}
1- n_1^\Gamma n_1^\Gamma & - n_1^\Gamma n_2^\Gamma & - n_1^\Gamma n_3^\Gamma\\
- n_2^\Gamma n_1^\Gamma & 1 - n_2^\Gamma n_2^\Gamma & - n_2^\Gamma n_3^\Gamma\\
- n_3^\Gamma n_1^\Gamma & - n_3^\Gamma n_2^\Gamma & 1 - n_3^\Gamma n_3^\Gamma
\end{pmatrix},
\end{equation*}
and the \emph{mean curvature $H_\Gamma$ in the direction $n_\Gamma$} by $H_\Gamma = H_\Gamma (x,t) = - {\rm{div}}_\Gamma n_\Gamma$, where $I_{3 \times 3}$ is the \emph{$3 \times 3$ identity matrix}, and $\otimes$ denotes the \emph{tensor product}. It is easy to check that $P_\Gamma n_\Gamma = { }^t (0,0,0)$ and $P_\Gamma \nabla f = \nabla_\Gamma f$.

Let us explain the key restrictions on the boundaries $\partial \Omega_T$ and $\Gamma_T$. We assume that
\begin{equation}\label{eq11}
\begin{cases}
v_B ={ }^t (0,0,0)  & \text{ on }\partial \Omega_T,\\
v_A \cdot n_\Gamma = v_B \cdot n_\Gamma = v_S \cdot n_\Gamma & \text{ on }\Gamma_{T},\\
P_\Gamma v_A = P_\Gamma v_B = r P_\Gamma v_S & \text{ on }\Gamma_T,
\end{cases}
\begin{cases}
( n_\Omega \cdot \nabla ) \theta_B =0  & \text{ on }\partial \Omega_T,\\
\theta_A = \theta_B = \theta_S & \text{ on }\Gamma_{T},
\end{cases}
\end{equation}
where $r \in \{ 0,1 \}$. We call $P_\Gamma v_A = P_\Gamma v_B = r P_\Gamma v_S$ a \emph{slip boundary condition} if $r=1$ and a \emph{no-slip boundary condition} if $r=0$. Note that we do not consider phase transition in this paper.

This paper has three purposes. The first purpose is to derive the following multiphase flow system with surface tension and flow:
\begin{equation}\label{eq12}
\begin{cases}
D_t^A \rho_A + ({\rm{div}} v_A) \rho_A = 0 & \text{ in }\Omega_{A,T},\\
D_t^B \rho_B + ({\rm{div}} v_B) \rho_B = 0 & \text{ in }\Omega_{B,T},\\
D_t^S \rho_S + ({\rm{div}}_\Gamma v_S) \rho_S = 0 & \text{ on }\Gamma_{T},
\end{cases}
\end{equation}
\begin{equation}\label{eq13}
\begin{cases}
\rho_A D_t^A e_A + ({\rm{div}} v_A) \pi_A = {\rm{div}} q_A + e_{D_A} & \text{ in }\Omega_{A,T},\\
\rho_B D_t^B e_B + ({\rm{div}} v_B ) \pi_B  = {\rm{div}} q_B +e_{D_B} & \text{ in }\Omega_{B,T},\\
\rho_S D_t^S e_S + ({\rm{div}}_\Gamma v_S) \pi_S = {\rm{div}}_\Gamma q_S + e_{D_S} + q_B \cdot n_\Gamma - q_A \cdot n_\Gamma & \text{ on }\Gamma_{T},
\end{cases}
\end{equation}
\begin{equation}\label{eq14}
\begin{cases}
\rho_A D_t^A v_A = {\rm{div}} \mathcal{T}_A  & \text{ in }\Omega_{A,T},\\
\rho_B D_t^B v_B = {\rm{div}} \mathcal{T}_B & \text{ in }\Omega_{B,T},\\
\rho_S D_t^S v_S = {\rm{div}}_\Gamma \mathcal{T}_S + \widetilde{\mathcal{T}}_B n_\Gamma - \widetilde{\mathcal{T}}_A n_\Gamma & \text{ on }\Gamma_{T},
\end{cases}
\end{equation}
where
\begin{equation}\label{eq15}
\begin{cases}
q_A = q_A (\theta_A ) := \kappa_A {\rm{grad}} \theta_A,\\ 
q_B = q_B (\theta_B ) := \kappa_B {\rm{grad}} \theta_B,\\ 
q_S = q_S (\theta_S ) := \kappa_S {\rm{grad}}_\Gamma \theta_S,
\end{cases}
\end{equation}
\begin{equation}\label{eq16}
\begin{cases}
e_{D_A} = e_{D_A} (v_A ) := \mu_A \vert D (v_A) \vert^2 + \lambda_A \vert {\rm{div}} v_A \vert^2,\\ 
e_{D_B} = e_{D_B} (v_B ) := \mu_B \vert D (v_B) \vert^2 + \lambda_B \vert {\rm{div}} v_B \vert^2,\\ 
e_{D_S} = e_{D_S} (v_S ) := \mu_S \vert D_\Gamma (v_S) \vert^2 + \lambda_S \vert {\rm{div}}_\Gamma v_S \vert^2, 
\end{cases}
\end{equation}
\begin{equation*}
\begin{cases}
D (v_A ) := \{ (\nabla v_A) + { }^t (\nabla v_A)  \}/2,\\
D (v_B ) := \{ (\nabla v_B) + { }^t (\nabla v_B)  \}/2,\\
D_\Gamma (v_S) := \{ (P_\Gamma \nabla_\Gamma v_S) + { }^t (P_\Gamma \nabla_\Gamma v_S) \}/2, 
\end{cases}
\end{equation*}
\begin{equation}\label{eq17}
\begin{cases}
\mathcal{T}_A = \mathcal{T}_A (v_A , \pi_A ) := \mu_A D (v_A) + \lambda_A ({\rm{div}} v_A) I_{3 \times 3} - \pi_A I_{3 \times 3},\\
\mathcal{T}_B = \mathcal{T}_B (v_B , \pi_B ) := \mu_B D (v_B) + \lambda_B ({\rm{div}} v_B) I_{3 \times 3} - \pi_B I_{3 \times 3},\\
\mathcal{T}_S = \mathcal{T}_S (v_S , \pi_S ) := \mu_S D_\Gamma (v_S) + \lambda_S ({\rm{div}}_\Gamma v_S) P_\Gamma - \pi_S P_\Gamma,
\end{cases}
\end{equation}
\begin{equation}\label{eq18}
\begin{cases}
\widetilde{\mathcal{T}}_A = \widetilde{\mathcal{T}}_A (v_A ,\pi_A ) = 
\begin{cases}
\mu_A  n_\Gamma \cdot (n_\Gamma \cdot \nabla ) v_A + \lambda_A ({\rm{div}} v_A) - \pi_A  \text{ if } r=0,\\
\mathcal{T}_A (v_A, \pi_A) \text{ if } r=1,
\end{cases}\\
\widetilde{\mathcal{T}}_B = \widetilde{\mathcal{T}}_B ( v_B , \pi_B ) = 
\begin{cases}
 \mu_B n_\Gamma \cdot (n_\Gamma \cdot \nabla ) v_B + \lambda_B ({\rm{div}} v_B) - \pi_B \text{ if }r=0,\\
\mathcal{T}_B (v_B , \pi_B) \text{ if } r =1.
\end{cases}
\end{cases}
\end{equation}
Here $\vert D (v_A) \vert^2 = D (v_A) : D(v_A)$, $\vert D (v_B) \vert^2 = D (v_B) : D (v_B)$, and $ \vert D_\Gamma (v_S) \vert^2 = D_\Gamma (v_S) : D_\Gamma (v_S)$. The symbol $:$ denotes the \emph{Frobenius inner product}, that is, $\mathcal{M} : \mathcal{N} =\sum_{i,j=1}^3 [\mathcal{M}]_{ij} [\mathcal{N}]_{ij}$, where $\mathcal{M}$, $\mathcal{N}$ are two $3\times 3$ matrices, and $[\mathcal{M}]_{ij}$ denotes the $(i,j)$-component of the matrix $\mathcal{M}$. We call $q_A$, $q_B$, $q_S$ the \emph{heat fluxes}, $e_{D_A}$, $e_{D_B}$, $e_{D_S}$ the \emph{energy densities for the energy dissipation due to the viscosities}, $D(v_A)$, $D (v_B)$ \emph{strain rate tensors}, $D_\Gamma (v_S)$ a \emph{surface strain tensor}, $\mathcal{T}_A$, $\mathcal{T}_B$ \emph{stress tensors}, and $\mathcal{T}_S$ a \emph{surface stress tensor}. We often call $\mathcal{T}_S$ the \emph{surface stress tensor determined by the Boussinesq-Scriven law}. More precisely, under the restrictions \eqref{eq11} we apply our energy densities and thermodynamic approaches to derive \eqref{eq12}-\eqref{eq14}. See Section \ref{sect4} for details.
\begin{remark}\label{rem12}
$(\rm{i})$ Using $n_\Gamma \cdot n_\Gamma =1$, $P_\Gamma \nabla f = \nabla_\Gamma f$, $(n_\Gamma \cdot \nabla_\Gamma )f = 0$, and $H_\Gamma = - {\rm{div}}_\Gamma n_\Gamma$, we easily check that
\begin{align*}
{\rm{div}}_\Gamma (\pi_S P_\Gamma ) & = {\rm{grad}}_\Gamma \pi_S + \pi_S H_\Gamma n_\Gamma,\\
2 D_\Gamma (v_S) & = P_\Gamma \{ (\nabla v_S) + { }^t (\nabla v_S) \} P_\Gamma = P_\Gamma \{ (\nabla_\Gamma v_S) + { }^t (\nabla_\Gamma v_S) \} P_\Gamma.
\end{align*}
Note that $2[D_\Gamma (v_S)]_{ij} = \partial^\Gamma_i v_j^S + \partial^\Gamma_j v_i^S - n^\Gamma_i (n_\Gamma \cdot \partial^\Gamma_j v_S) -n^\Gamma_j (n_\Gamma \cdot \partial^\Gamma_i v_S)$,
\begin{equation*}\nabla v_S =
\begin{pmatrix}
\partial_1 v_1^S & \partial_2 v_1^S & \partial_3 v_1^S\\
\partial_1 v_2^S & \partial_2 v_2^S & \partial_3 v_2^S\\
\partial_1 v_3^S & \partial_2 v_3^S & \partial_3 v_3^S
\end{pmatrix},{ \ } \nabla_\Gamma v_S =
\begin{pmatrix}
\partial^\Gamma_1 v_1^S & \partial^\Gamma_2 v_1^S & \partial^\Gamma_3 v_1^S\\
\partial^\Gamma_1 v_2^S & \partial^\Gamma_2 v_2^S & \partial^\Gamma_3 v_2^S\\
\partial^\Gamma_1 v_3^S & \partial^\Gamma_2 v_3^S & \partial^\Gamma_3 v_3^S
\end{pmatrix}.
\end{equation*}
We often call $\pi_S H_\Gamma n_\Gamma$ \emph{surface tension}.\\
$(\rm{ii})$ If the fluids in $\Omega_{A,T}$, $\Omega_{B,T}$, $\Gamma_T$ are barotropic fluids, then we can write
\begin{equation*}
\begin{cases}
\pi_A = \pi_A ( \rho_A) =  \rho_A p_A' (\rho_A) - p_A (\rho_A),\\
\pi_B = \pi_B ( \rho_B) =  \rho_B p_B' (\rho_B) - p_B (\rho_B ),\\
\pi_S = \pi_S ( \rho_S) =  \rho_S p_S' (\rho_S) - p_S (\rho_S ).
\end{cases}
\end{equation*}
Here $p_A$, $p_B$, $p_S$ are three $C^1$-functions, $p' = p'(r) = {dp}/{dr}(r)$. See Theorem \ref{thm24}, Remark \ref{rem25}, and Section \ref{sect4} for details.
\end{remark}

The second purpose is to study the conservative forms and conservation laws of system \eqref{eq12}-\eqref{eq14}. In fact, if we set $D_t^N f = \partial_t f + (v_S \cdot n_\Gamma) (n_\Gamma \cdot \nabla )f $, and the \emph{total energy} $E_\sharp = E_\sharp (x,t)$ by $E_\sharp =\rho_\sharp \vert v_\sharp \vert^2/2 + \rho_\sharp e_\sharp$, then we can write our system as the conservative form:
\begin{equation}\label{eq19}
\begin{cases}
\partial_t \rho_A + {\rm{div}}( \rho_A v_A ) = 0 & \text{ in }\Omega_{A,T},\\
\partial_t \rho_B + {\rm{div}} ( \rho_B v_B ) = 0 & \text{ in }\Omega_{B,T},\\
D_t^N \rho_S + {\rm{div}}_\Gamma ( \rho_S v_S) = 0 & \text{ on }\Gamma_{T},
\end{cases}
\end{equation}
\begin{equation}\label{eq1010}
\begin{cases}
\partial_t E_A + {\rm{div}} (E_A v_A - q_A - \mathcal{T}_A v_A)   = 0 & \text{ in }\Omega_{A,T},\\
\partial_t E_B + {\rm{div}} (E_B v_B - q_B - \mathcal{T}_B v_B)   = 0 & \text{ in }\Omega_{B,T},\\
D_t^N E_S + {\rm{div}}_\Gamma (E_S v_S - q_S - \mathcal{T}_S v_S)   = \mathcal{E}_S & \text{ on }\Gamma_{T},
\end{cases}
\end{equation}
\begin{equation}\label{eq1011}
\begin{cases}
\partial_t (\rho_A v_A ) + {\rm{div}}( \rho_A v_A \otimes v_A - \mathcal{T}_A) = { }^t (0,0,0)  & \text{ in }\Omega_{A,T},\\
\partial_t (\rho_B v_B ) + {\rm{div}}( \rho_B v_B \otimes v_B - \mathcal{T}_B ) = { }^t ( 0,0,0) & \text{ in }\Omega_{B,T},\\
D_t^N ( \rho_S v_S ) + {\rm{div}}_\Gamma ( \rho_S v_S \otimes v_S - \mathcal{T}_S ) = \widetilde{\mathcal{T}}_B n_\Gamma - \widetilde{\mathcal{T}}_A n_\Gamma & \text{ on }\Gamma_{T},
\end{cases}
\end{equation}
where
\begin{equation*}
\mathcal{E}_S := \widetilde{\mathcal{T}}_B n_\Gamma \cdot v_S - \widetilde{\mathcal{T}}_A n_\Gamma \cdot v_S + q_B \cdot n_\Gamma - q_A \cdot n_\Gamma .
\end{equation*}
Moreover, any solution to system \eqref{eq11}-\eqref{eq14} satisfies that for $t_1< t_2$,
\begin{multline}\label{eq1012}
\int_{\Omega_A (t_2 )} \rho_A (x,t_2) { \ }d x + \int_{ \Omega_B (t_2)} \rho_B (x,t_2) { \ } d x + \int_{\Gamma (t_2)} \rho_S (x,t_2) { \ } d \mathcal{H}^2_x\\
= \int_{\Omega_A (t_1)} \rho_A (x,t_1) { \ }d x + \int_{ \Omega_B (t_1)} \rho_B (x,t_1) { \ } d x + \int_{\Gamma (t_1)} \rho_S (x,t_1) { \ } d \mathcal{H}^2_x,
\end{multline}
\begin{multline}\label{eq1013}
\int_{\Omega_A (t_2 )} E_A { \ }d x + \int_{ \Omega_B (t_2)} E_B { \ } d x + \int_{\Gamma (t_2)} E_S { \ } d \mathcal{H}^2_x\\
= \int_{\Omega_A (t_1)} E_A { \ }d x + \int_{ \Omega_B (t_1)} E_B { \ } d x + \int_{\Gamma (t_1)} E_S { \ } d \mathcal{H}^2_x,
\end{multline}
\begin{multline}\label{eq1014}
\int_{\Omega_A (t_2 )} \frac{1}{2} \rho_A \vert v_A \vert^2 { \ }d x + \int_{ \Omega_B (t_2)} \frac{1}{2} \rho_B \vert v_B \vert^2 { \ } d x + \int_{\Gamma (t_2)} \frac{1}{2} \rho_S \vert v_S \vert^2 { \ } d \mathcal{H}^2_x\\
+ \int_{t_1}^{t_2} \int_{\Omega_A (t)} e_{D_A} { \ }d x d t + \int_{t_1}^{t_2} \int_{ \Omega_B (t)} e_{D_B} { \ } d x d t + \int_{t_1}^{t_2} \int_{\Gamma (t)} e_{D_S} { \ } d \mathcal{H}^2_x d t\\
= \int_{\Omega_A (t_1 )} \frac{1}{2} \rho_A \vert v_A \vert^2 { \ }d x + \int_{ \Omega_B (t_1)} \frac{1}{2} \rho_B \vert v_B \vert^2 { \ } d x + \int_{\Gamma (t_1)} \frac{1}{2} \rho_S \vert v_S \vert^2 { \ } d \mathcal{H}^2_x\\
+ \int_{t_1}^{t_2} \int_{\Omega_A (t)} ({\rm{div}} v_A) \pi_A { \ }d x d t + \int_{t_1}^{t_2} \int_{ \Omega_B (t)} ({\rm{div}} v_B) \pi_B { \ } d x d t\\
 + \int_{t_1}^{t_2} \int_{\Gamma (t)} ({\rm{div}}_\Gamma v_S) \pi_S { \ } d \mathcal{H}^2_x d t.
\end{multline}
Here $d \mathcal{H}^2_x$ denotes the \emph{2-dimensional Hausdorff measure}. Under some assumptions (see Theorem \ref{thm28}), any solution to system \eqref{eq11}-\eqref{eq13} satisfies that for $t_1 < t_2$
\begin{multline}\label{eq1015}
\int_{\Omega_A (t_2 )} \rho_A v_A { \ }d x + \int_{ \Omega_B (t_2)} \rho_B v_B { \ } d x + \int_{\Gamma (t_2)} \rho_S v_S { \ } d \mathcal{H}^2_x\\
= \int_{\Omega_A (t_1)} \rho_A v_A { \ }d x + \int_{ \Omega_B (t_1)} \rho_B v_B { \ } d x + \int_{\Gamma (t_1)} \rho_S v_S { \ } d \mathcal{H}^2_x.
\end{multline}
We often call \eqref{eq1012}, \eqref{eq1013}, \eqref{eq1014}, and \eqref{eq1015}, the \emph{law of conservation of mass}, the \emph{law of conservation of total energy}, the \emph{energy law of our system}, and the \emph{law of conservation of momentum}, respectively. See Theorem \ref{thm28} and Section \ref{sect5} for details.
\begin{remark}\label{rem13}
From $D_t^S f = D_t^N f + (v_S \cdot \nabla_\Gamma )f $ for $f \in C^1 (\mathbb{R}^4)$, we see that
\begin{align*}
& D_t^N ( \rho_S f ) + {\rm{div}}_\Gamma ( \rho_S f v_S ) = \rho_S D_t^S f,\\
& D_t^N ( \rho_S v_S ) + {\rm{div}}_\Gamma ( \rho_S v_S \otimes v_S ) = \rho_S D_t^S v_S.
\end{align*}
Since $\rho_S, v_S \in C^1 (\mathbb{R}^4)$ in this paper, we can define $D_t^N$ for $\rho_S$, $\rho_S v_S$.
\end{remark}

The third purpose is to investigate the thermodynamic potential such as the enthalpy $h_\sharp$, the entropy $\varsigma_\sharp$, the Helmholtz free energy $F_\sharp^H$, and the Gibbs free energy $F_\sharp^G$ of the fluid in $\Omega_\sharp (t)$, where $\sharp = A,B,S$. Assume that $(\rho_\sharp, \theta_\sharp)$ are positive functions. Set the \emph{enthalpy} $h_\sharp$ by $h_\sharp = e_\sharp + {\pi_\sharp}/{\rho_\sharp}$. Then
\begin{equation}\label{eq1016}
\begin{cases}
\partial_t (\rho_A h_A) + {\rm{div}} (\rho_A h_A v_A - q_A )  = e_{D_A} + D_t^A \pi_A,\\
\partial_t (\rho_B h_B) + {\rm{div}} (\rho_B h_B v_B - q_B )  = e_{D_B} + D_t^B \pi_B,\\
D_t^N (\rho_S h_S) + {\rm{div}}_\Gamma (\rho_S h_S v_S - q_S ) = e_{D_S} + D_t^S \pi_S + q_B \cdot n_\Gamma - q_A \cdot n_\Gamma.
\end{cases}
\end{equation}
Suppose that the thermodynamic identity (Gibbs \cite{Gib1906}): $D_t^\sharp e_\sharp = \theta_\sharp D_t^\sharp \varsigma_\sharp - \pi_\sharp D_t^\sharp (1/{\rho_\sharp})$ holds. Then
\begin{equation}\label{eq1017}
\begin{cases}
\partial_t (\rho_A \varsigma_A ) + {\rm{div}} \bigg( \rho_A \varsigma_A v_A - \frac{q_A}{\theta_A} \bigg) = \frac{e_{D_A}}{\theta_A} + \frac{q_A \cdot {\rm{grad}} \theta_A  }{\theta_A^2},\\
\partial_t (\rho_B \varsigma_B ) + {\rm{div}} \bigg( \rho_B \varsigma_B v_B - \frac{q_B}{\theta_B} \bigg) = \frac{e_{D_B}}{\theta_B} + \frac{q_B \cdot {\rm{grad}} \theta_B  }{\theta_B^2},\\
D_t^N (\rho_S \varsigma_S ) + {\rm{div}}_\Gamma \bigg( \rho_S \varsigma_S v_S - \frac{q_S}{\theta_S} \bigg) = \frac{e_{D_S}}{\theta_S} + \frac{q_S \cdot {\rm{grad}}_\Gamma \theta_S  }{\theta_S^2} + \frac{q_B \cdot n_\Gamma - q_A \cdot n_\Gamma   }{\theta_S}.
\end{cases}
\end{equation}
Set the \emph{Helmholtz free energy} $F^H_\sharp$ by $F^H_\sharp = e_\sharp - \theta_\sharp \varsigma_\sharp$. Then
\begin{equation}\label{eq1018}
\begin{cases}
\rho_A D_t^A F^H_A + \rho_A \varsigma_A D^A_t \theta_A = - ({\rm{div}} v_A) \pi_A,\\
\rho_B D_t^B F^H_B + \rho_B \varsigma_B D^B_t \theta_B = - ({\rm{div}} v_B) \pi_B,\\
\rho_S D_t^S F^H_S + \rho_S \varsigma_S D^S_t \theta_S = - ({\rm{div}}_\Gamma v_S) \pi_S.
\end{cases}
\end{equation}
Set the \emph{Gibbs free energy} $F^G_\sharp$ by $F^G_\sharp = h_\sharp - \theta_\sharp \varsigma_\sharp$. Then
\begin{equation}\label{eq1019}
\begin{cases}
\rho_A D_t^A F^G_A + \rho_A \varsigma_A D^A_t \theta_A = D_t^A \pi_A,\\
\rho_B D_t^B F^G_B + \rho_B \varsigma_B D^B_t \theta_B = D_t^B \pi_B,\\
\rho_S D_t^S F^G_S + \rho_S \varsigma_S D^S_t \theta_S = D_t^S \pi_S.
\end{cases}
\end{equation}
See Theorem \ref{thm29} and Section \ref{sect6} for details.

\begin{remark}\label{rem14} $(\rm{i})$ Since
\begin{align*}
\mathcal{T}_A (v_A , \pi_A ) : D (v_A) = e_{D_A} - ({\rm{div}} v_A) \pi_A,\\
\mathcal{T}_B (v_B , \pi_B ) : D (v_B) = e_{D_B} - ({\rm{div}} v_B) \pi_B,\\
\mathcal{T}_S (v_S , \pi_S ) : D_\Gamma (v_S) = e_{D_S} - ({\rm{div}}_\Gamma v_S) \pi_S,
\end{align*}
it follows from \eqref{eq1018} to see that
\begin{equation*}
\begin{cases}
\rho_A D_t^A F^H_A + \rho_A \varsigma_A D^A_t \theta_A - \mathcal{T}_A(v_A, \pi_A) : D (v_A) = - e_{D_A},\\
\rho_B D_t^B F^H_B + \rho_B \varsigma_B D^B_t \theta_B - \mathcal{T}_B(v_B,\pi_B) : D (v_B) = - e_{D_B},\\
\rho_S D_t^S F^H_S + \rho_S \varsigma_S D^S_t \theta_S - \mathcal{T}_S (v_S, \pi_S) : D_\Gamma (v_S) = - e_{D_S}.
\end{cases}
\end{equation*}
Note that $P_\Gamma : D_\Gamma (v_S) = {\rm{div}}_\Gamma v_S$.\\
\noindent $(\rm{ii})$ Set $\mathcal{D}_A f := \rho_A D_t^A f$, $\mathcal{D}_B f := \rho_B D_t^B f$, $\mathcal{D}_S f := \rho_S D_t^S f$. Then
\begin{equation*}
\begin{cases}
\mathcal{D}_A e_A = \theta_A \mathcal{D}_A \varsigma_A - \pi_A \mathcal{D}_A (1/{\rho_A}),\\
\mathcal{D}_B e_B = \theta_B \mathcal{D}_B \varsigma_B - \pi_B \mathcal{D}_B (1/{\rho_B}),\\
\mathcal{D}_S e_S = \theta_S \mathcal{D}_S \varsigma_S - \pi_S \mathcal{D}_S (1/{\rho_S}),
\end{cases}
\begin{cases}
\mathcal{D}_A h_A = \theta_A \mathcal{D}_A \varsigma_A + (1/{\rho_A}) \mathcal{D}_A \pi_A,\\
\mathcal{D}_B h_B = \theta_B \mathcal{D}_B \varsigma_B + (1/{\rho_B}) \mathcal{D}_B \pi_B,\\
\mathcal{D}_S h_S = \theta_S \mathcal{D}_S \varsigma_S + (1/{\rho_S}) \mathcal{D}_S \pi_S,
\end{cases}
\end{equation*}
\begin{equation*}
\begin{cases}
\mathcal{D}_A F^H_A = - \varsigma_A \mathcal{D}_A \theta_A - \pi_A  \mathcal{D}_A(1/{\rho_A}),\\
\mathcal{D}_B F^H_B = - \varsigma_B \mathcal{D}_B \theta_B - \pi_B  \mathcal{D}_B(1/{\rho_B}),\\
\mathcal{D}_S F^H_S = - \varsigma_S \mathcal{D}_S \theta_S - \pi_S  \mathcal{D}_S(1/{\rho_S}),
\end{cases}
\begin{cases}
\mathcal{D}_A F^G_A = - \varsigma_A \mathcal{D}_A \theta_A + (1/{\rho_A})  \mathcal{D}_A \pi_A,\\
\mathcal{D}_B F^G_B = - \varsigma_B \mathcal{D}_B \theta_B + (1/{\rho_B})  \mathcal{D}_B \pi_B,\\
\mathcal{D}_S F^G_S = - \varsigma_S \mathcal{D}_S \theta_S + (1/{\rho_S}) \mathcal{D}_S \pi_S.
\end{cases}
\end{equation*}
We often call $1/{\rho_\sharp}$ a \emph{specific volume}.
\end{remark}

Let us explain the main difficulties in the derivation of our multiphase flow system with surface tension and flow, and the key ideas to overcome these difficulties. The main difficulties are to derive the viscous terms of the system,  to derive the surface tension from a theoretical point of view, and to derive the dependencies among fluid-flows in two moving domains and surface flow. To overcome these difficulties, we apply the first law of thermodynamics (Theorem \ref{thm24}), our energy densities (Definition \ref{def22}), and the conservation law of total energy to derive equations \eqref{eq13} and \eqref{eq14}. See Section \ref{sect4} for details.

Let us mention the study of surface flow (interfacial flow). Boussinesq \cite{Bou13} first discovered the existence of surface flow. Scriven \cite{Scr60} considered their surface stress tensor. Slattery \cite{Sla64} investigated some properties of the surface stress tensor determined by the Boussinesq-Scriven law (see $\mathcal{T}_S$ in \eqref{eq17}). Then many researchers have studied surface flow (see Slattery-Sagis-Oh \cite{SSO07} and Gatignol-Prud'homme \cite{GP01} for the study of interfacial phenomena).

Let us state derivations of the governing equations for the motion of the viscous fluid on manifolds and surfaces. Taylor \cite{Tay92} introduced their surface stress tensor to make their incompressible viscous fluid system on a manifold. Mitsumatsu-Yano \cite{MY02} applied their energetic variational approach to derive their incompressible viscous fluid system on a manifold. Arnaudon-Cruzeiro \cite{AC12} made use of their stochastic variational approach to derive their incompressible viscous fluid system on a manifold. Koba-Liu-Giga \cite{KLG17} employed their energetic variational approach and the generalized Helmholtz-Weyl decomposition on a closed surface to derive their incompressible fluid systems on an evolving closed surface. Koba \cite{K18,K22} applied their energetic variational approaches and the first law of thermodynamics to derive their compressible fluid flow systems on an evolving closed surface and an evolving surface with a boundary. This paper modifies and improves the methods in \cite{K18,K22} to derive our multiphase flow system.

Now we mention results for modeling of multiphase flow system with surface flow. Bothe-Pr\"{u}ss \cite{BP10} made their multiphase flow system with surface flow by using the surface stress tensor determined by the Boussinesq-Scriven law. Koba \cite{K23} derived the inviscid multiphase flow system with surface flow by applying a geometric variational approach. This paper derives our multiphase flow system from a thermodynamic point of view. Therefore, our modeling methods are different from ones in \cite{BP10} and \cite{K23}.

Finally, we introduce some results and textbooks related to this paper. Hyon-Kwak-Liu \cite{HKL10} and Koba-Sato \cite{KS17} applied their energetic variational approaches to derive and study their complex and non-Newtonian fluid systems in domains. Feireisl \cite{Fei15} studied the motion of the viscous fluid in a domain from a thermodynamic point of view. We refer the readers to Gyarmati \cite{Gya70} and Gurtin-Fried-Anand \cite{GFA10} for the theory of thermodynamics, Chapter XIII in Angel \cite{Ang06} for thermodynamical potential such as internal energy, enthalpy, entropy, and free energies, and Pr\"{u}ss-Simonett \cite{PS16} for several elliptic and parabolic equations on hypersurfaces.

The outline of this paper is as follows: In Section \ref{sect2}, we first introduce the transport theorems and the energy densities for our model, and then we state the main results of this paper. In Section \ref{sect3}, we make use of the transport theorems to derive the first law of thermodynamics, and apply integration by parts to calculate variations of our dissipation energies. In Section \ref{sect4}, we apply our thermodynamic approaches to make mathematical models for multiphase flow with surface tension and flow. In Section \ref{sect5}, we study the conservation and energy laws of our system. In Section \ref{sect6}, we investigate the thermodynamic potential for our system.

\section{Main Results}\label{sect2}

We first introduce the transport theorems and the energy densities for our multiphase flow system. Then we state the main results.
\begin{definition}[$\Omega_T$ is flowed by the velocity fields $(v_A,v_B,v_S)$]\label{def21} We say that $\Omega_T$ is \emph{flowed by the velocity fields} $(v_A,v_B ,v_S)$ if for each $0< t <T$, $f \in C^1 (\mathbb{R}^4)$, and $\Lambda \subset \Omega$,
\begin{align}
\label{eq21} \frac{d}{d t} \int_{\Omega_A (t) \cap \Lambda} f (x,t) { \ }d x & = \int_{\Omega_A (t) \cap \Lambda} \{ D_t^A f + ({\rm{div}} v_A) f \} { \ }dx,\\ 
\label{eq22} \frac{d}{d t} \int_{\Omega_B (t) \cap \Lambda} f (x,t) { \ }d x & = \int_{\Omega_B (t) \cap \Lambda} \{ D_t^B f + ({\rm{div}} v_B ) f \} { \ }dx,\\ 
\label{eq23} \frac{d}{d t} \int_{\Gamma (t) \cap \Lambda} f (x,t) { \ }d \mathcal{H}^2_x & = \int_{\Gamma (t) \cap \Lambda} \{ D_t^S f + ({\rm{div}}_\Gamma v_S) f \} { \ }d \mathcal{H}^2_x. 
\end{align}
Here $D_t^\sharp f = \partial_t f + (v_\sharp \cdot \nabla) f$, ${\rm{div}}_\Gamma v_S = \partial_1^\Gamma v^S_1 + \partial_2^\Gamma v^S_2 + \partial_3^\Gamma v^S_3$, $\partial_j^\Gamma f = \partial_j f - n^\Gamma_j (n_\Gamma \cdot \nabla )f $, where $\sharp = A,B,S$, and $j=1,2,3$.
\end{definition}
\noindent We often call \eqref{eq21}, \eqref{eq22} the \emph{transport theorems}, and \eqref{eq22} the \emph{surface transport theorem}. The derivation of the surface transport theorem can be founded in \cite{Bet86,GSW89,DE07,KLG17}. Throughout this paper we assume that $\Omega_{T}$ is flowed by the velocity fields $(v_A, v_B ,v_S)$.

\begin{definition}[Energy densities]\label{def22} Set
\begin{equation*}
\begin{cases}
e_{K_A} = \rho_A \vert v_A \vert^2/2,\\ 
e_{K_B} = \rho_B \vert v_B \vert^2/2,\\ 
e_{K_S} = \rho_S \vert v_S \vert^2/2, 
\end{cases}
\begin{cases}
e_{D_A} = \mu_A \vert D (v_A) \vert^2 + \lambda_A \vert {\rm{div}} v_A \vert^2,\\ 
e_{D_B} = \mu_B \vert D (v_B) \vert^2 + \lambda_B \vert {\rm{div}} v_B \vert^2,\\ 
e_{D_S} = \mu_S \vert D_\Gamma (v_S) \vert^2 + \lambda_S \vert {\rm{div}}_\Gamma v_S \vert^2, 
\end{cases}
\end{equation*}
\begin{equation*}
\begin{cases}
e_{W_A} = ({\rm{div}} v_A ) \pi_A,\\ 
e_{W_B} =  ( {\rm{div}} v_B ) \pi_B,\\ 
e_{W_S} = ( {\rm{div}}_\Gamma v_S ) \pi_S, 
\end{cases}
\begin{cases}
e_{Q_A} = \kappa_A \vert {\rm{grad}} \theta_A \vert^2,\\ 
e_{Q_B} = \kappa_B \vert {\rm{grad}} \theta_B \vert^2,\\ 
e_{Q_S} =  \kappa_S \vert {\rm{grad}}_\Gamma \theta_S \vert^2.
\end{cases}
\end{equation*}
We call $e_{K_\sharp}$ the \emph{kinetic energy}, $e_{D_\sharp}$ the \emph{energy density for the energy dissipation due to the viscosities} $(\mu_\sharp, \lambda_\sharp )$, $e_{W_\sharp}$ the \emph{power density for the work done by the pressure} $\pi_\sharp$, and $e_{Q_\sharp}$ the \emph{energy density for the energy dissipation due to thermal diffusion}. 
\end{definition}
\noindent See \cite{KS17}, \cite{K18}, and Remark 2.5 in \cite{K22} for mathematical validity of the energy densities. Applying the energy densities, the restrictions \eqref{eq11}, and our thermodynamic approaches, we derive system \eqref{eq12}-\eqref{eq14} in Section \ref{sect4}.

We now state the main results of this paper. From Definition \ref{def21}, we have
\begin{proposition}[Continuity equations]\label{prop23}
Assume that for each $0< t <T$ and $\Lambda \subset \Omega$,
\begin{align*}
\frac{d}{d t} \bigg( \int_{\Omega_A (t) \cap \Lambda} \rho_A (x,t) { \ }d x + \int_{\Omega_B (t) \cap \Lambda} \rho_B (x,t) { \ }d x +  \int_{\Gamma (t) \cap \Lambda} \rho_S (x,t) { \ }d \mathcal{H}^2_x  \bigg) = 0. 
\end{align*}
Then $( \rho_A , \rho_B , \rho_S )$ satisfies \eqref{eq12}.
\end{proposition}
\noindent The proof of Proposition \ref{prop23} is left for the readers.

\begin{theorem}[First law of thermodynamics]\label{thm24}
Let $\widetilde{Q}_A, \widetilde{Q}_B, \widetilde{Q}_S \in C (\mathbb{R}^4)$. Assume that $(\rho_A , \rho_B , \rho_S )$ satisfies \eqref{eq12}. Then\\
\noindent $(\rm{i})$ Suppose that for every $0 < t <T$ and $\Lambda \subset \Omega$,
\begin{align*}
\frac{d}{d t} \int_{\Omega_A (t) \cap \Lambda} \rho_A e_A { \ }d x = \int_{\Omega_A (t) \cap \Lambda} \{ \widetilde{Q}_A - ({\rm{div}} v_A) \pi_A \} { \ } d x,\\
\frac{d}{d t} \int_{\Omega_B (t) \cap \Lambda} \rho_B e_B { \ }d x = \int_{\Omega_B (t) \cap \Lambda} \{ \widetilde{Q}_B - ({\rm{div}} v_B) \pi_B \} { \ } d x,\\
\frac{d}{dt} \int_{\Gamma (t) \cap \Lambda} \rho_S e_S { \ }d \mathcal{H}^2_x = \int_{\Gamma (t) \cap \Lambda} \{ \widetilde{Q}_S - ({\rm{div}}_\Gamma v_S) \pi_S \} { \ } d \mathcal{H}^2_x.
\end{align*}
Then
\begin{equation}\label{eq24}
\begin{cases}
\rho_A D_t^A e_A + ({\rm{div}} v_A ) \pi_A = \widetilde{Q}_A \text{ in } \Omega_{A,T},\\
\rho_B D_t^B e_B + ({\rm{div}} v_B ) \pi_B = \widetilde{Q}_B \text{ in } \Omega_{B,T},\\
\rho_S D_t^S e_S + ({\rm{div}}_\Gamma v_S ) \pi_S = \widetilde{Q}_S \text{ on } \Gamma_T.
\end{cases}
\end{equation}
\noindent $(\rm{ii})$ Let $p_A,p_B,p_S \in C^1 (\mathbb{R})$. Suppose that for every $0 < t <T$ and $\Lambda \subset \Omega$,
\begin{align*}
\frac{d}{d t} \int_{\Omega_A (t) \cap \Lambda} \{ \rho_A e_A  - p_A (\rho_A ) \} { \ }d x = \int_{\Omega_A (t) \cap \Lambda} \widetilde{Q}_A { \ } d x,\\
\frac{d}{d t} \int_{\Omega_B (t) \cap \Lambda} \{ \rho_B e_B  - p_B (\rho_B ) \} { \ }d x = \int_{\Omega_B (t) \cap \Lambda} \widetilde{Q}_B { \ } d x,\\
\frac{d}{dt} \int_{\Gamma (t) \cap \Lambda} \{ \rho_S e_S  - p_S (\rho_S ) \} { \ }d \mathcal{H}^2_x = \int_{\Gamma (t) \cap \Lambda} \widetilde{Q}_S { \ } d \mathcal{H}^2_x.
\end{align*}
Then
\begin{equation}\label{eq25}
\begin{cases}
\rho_A D_t^A e_A + ({\rm{div}} v_A ) \Pi_A = \widetilde{Q}_A \text{ in } \Omega_{A,T},\\
\rho_B D_t^B e_B + ({\rm{div}} v_B ) \Pi_B = \widetilde{Q}_B \text{ in } \Omega_{B,T},\\
\rho_S D_t^S e_S + ({\rm{div}}_\Gamma v_S ) \Pi_S = \widetilde{Q}_S \text{ on } \Gamma_T.
\end{cases}
\end{equation}
Here
\begin{equation}\label{eq26}
\begin{cases}
\Pi_A = \Pi_A (\rho_A ) = \rho_A p'_A (\rho_A) - p_A (\rho_A),\\
\Pi_B = \Pi_B (\rho_B ) = \rho_B p'_B (\rho_B) - p_B (\rho_B),\\
\Pi_S = \Pi_S (\rho_S ) = \rho_S p'_S (\rho_S) - p_S (\rho_S).
\end{cases}
\end{equation}
\end{theorem}

\begin{remark}\label{rem25}$(\rm{i})$ We can write system \eqref{eq24} as follows:
\begin{equation*}
\begin{cases}
\mathcal{D}_A e_A = \widetilde{Q}_A - \pi_A \mathcal{D}_A (1/\rho_A),\\
\mathcal{D}_B e_B = \widetilde{Q}_B - \pi_B \mathcal{D}_B (1/\rho_B),\\
\mathcal{D}_S e_S = \widetilde{Q}_S - \pi_S \mathcal{D}_S (1/\rho_S),
\end{cases}
\end{equation*}
where $\mathcal{D}_\sharp f = \rho_\sharp D_t^\sharp f $. Therefore, we call Theorem \ref{thm24} the \emph{first law of thermodynamics} in this paper. See also $(\rm{ii})$ in Remark \ref{rem14}.\\
$(\rm{ii})$ The pressures $(\Pi_A, \Pi_B, \Pi_S)$ derived from the assertion $(\rm{ii})$ of Theorem \ref{thm24} correspond to the pressures derived from an energetic variational approach $($see \cite{K23}$)$.
\end{remark}

Next we consider the variation of our dissipation energies. Let $r \in \{ 0, 1 \}$ and $0<t<T$. Let $\varphi_A, \varphi_B, \varphi_S \in [C^\infty (\mathbb{R}^3)]^3$ and $\psi_A , \psi_B, \psi_S \in C^\infty ( \mathbb{R}^3)$. For $- 1 < \varepsilon < 1$, $v_A^\varepsilon := v_A + \varepsilon \varphi_A$, $v_B^\varepsilon := v_B + \varepsilon \varphi_B$, $v_S^\varepsilon := v_S + \varepsilon \varphi_S$, $\theta^\varepsilon_A := \theta_A + \varepsilon \psi_A$, $\theta^\varepsilon_B := \theta_B + \varepsilon \psi_B$, and $\theta^\varepsilon_S := \theta_S + \varepsilon \psi_S$. We call ($v_A^\varepsilon$, $v_B^\varepsilon$, $v_S^\varepsilon$, $\theta_A^\varepsilon$, $\theta_B^\varepsilon$, $\theta_S^\varepsilon$) variations of ($v_A$, $v_B$, $v_S$, $\theta_A$, $\theta_B$, $\theta_S$). From \eqref{eq11}, for $-1 < \varepsilon < 1$, we assume that
\begin{equation}\label{eq27}
\begin{cases}
v^\varepsilon_B ={ }^t (0,0,0)  & \text{ on }\partial \Omega,\\
v^\varepsilon_A \cdot n_\Gamma = v^\varepsilon_B \cdot n_\Gamma = v^\varepsilon_S \cdot n_\Gamma & \text{ on }\Gamma (t),\\
P_\Gamma v^\varepsilon_A = P_\Gamma v^\varepsilon_B = r P_\Gamma v^\varepsilon_S & \text{ on }\Gamma (t),
\end{cases}
\begin{cases}
( n_\Omega \cdot \nabla ) \theta^\varepsilon_B =0  & \text{ on }\partial \Omega,\\
\theta^\varepsilon_A = \theta^\varepsilon_B = \theta^\varepsilon_S & \text{ on }\Gamma (t).
\end{cases}
\end{equation}
Then we have
\begin{equation}\label{eq28}
\begin{cases}
\varphi_B ={ }^t (0,0,0)  & \text{ on }\partial \Omega,\\
\varphi_A \cdot n_\Gamma = \varphi_B \cdot n_\Gamma = \varphi_S \cdot n_\Gamma & \text{ on }\Gamma (t),\\
P_\Gamma \varphi_A = P_\Gamma \varphi_B = r P_\Gamma \varphi_S & \text{ on }\Gamma (t),
\end{cases}
\end{equation}
and
\begin{equation}\label{eq29}
\begin{cases}
( n_\Omega \cdot \nabla ) \psi_B =0  & \text{ on }\partial \Omega,\\
\psi_A = \psi_B = \psi_S & \text{ on }\Gamma (t).
\end{cases}
\end{equation}
For each variation ($v_A^\varepsilon$, $v_B^\varepsilon$, $v_S^\varepsilon$, $\theta_A^\varepsilon$, $\theta_B^\varepsilon$, $\theta_S^\varepsilon$),
\begin{multline*}
E_D [v_A^\varepsilon, v_B^\varepsilon , v_S^\varepsilon ] := \int_{\Omega_A (t)} \bigg( -\frac{\mu_A}{2} \vert D (v_A^\varepsilon ) \vert^2 - \frac{\lambda_A}{2} \vert {\rm{div}} v_A^\varepsilon \vert^2 \bigg) { \ } dx\\
+ \int_{\Omega_B (t)} \bigg( -\frac{\mu_B}{2} \vert D (v_B^\varepsilon ) \vert^2 - \frac{\lambda_B}{2} \vert {\rm{div}} v_B^\varepsilon \vert^2 \bigg) { \ } dx\\
+ \int_{\Gamma (t)} \bigg( -\frac{\mu_S}{2} \vert D_\Gamma (v_S^\varepsilon ) \vert^2 - \frac{\lambda_S}{2} \vert {\rm{div}}_\Gamma v_S^\varepsilon \vert^2 \bigg) { \ } d \mathcal{H}^2_x,
\end{multline*}
\begin{multline*}
E_W [v_A^\varepsilon , v_B^\varepsilon , v_S^\varepsilon ] := \int_{\Omega_A (t)} ({\rm{div}} v_A^\varepsilon ) \pi_A { \ } d x\\
+ \int_{\Omega_B (t)} ( {\rm{div}} v_B^\varepsilon ) \pi_B { \ } d x + \int_{\Gamma (t)} ( {\rm{div}}_\Gamma v_S^\varepsilon ) \pi_S { \ } d \mathcal{H}^2_x,
\end{multline*}
and
\begin{multline*}
E_{TD}[ \theta_A^\varepsilon , \theta_B^\varepsilon , \theta_S^\varepsilon ] := \int_{\Omega_A (t)} \bigg( - \frac{\kappa_A}{2} \vert {\rm{grad}} \theta_A^\varepsilon \vert^2 \bigg) { \ }d x\\
+ \int_{\Omega_B (t)} \bigg( - \frac{\kappa_B}{2} \vert {\rm{grad}} \theta_B^\varepsilon \vert^2 \bigg) { \ }d x + \int_{\Gamma (t)} \bigg( - \frac{\kappa_S}{2} \vert {\rm{grad}}_\Gamma \theta_S^\varepsilon \vert^2 \bigg) { \ }d \mathcal{H}^2_x.
\end{multline*}
Set $E_{D+W}[\cdot] = E_D[\cdot] + E_W[\cdot]$. We call $E_D$ the \emph{energy dissipation due to viscosities}, $E_W$ the \emph{work done by pressures}, and $E_{TD}$ the \emph{energy dissipation due to thermal diffusion}.

\begin{theorem}[Forces derived from variation of dissipation energies]\label{thm26}
Let $r \in \{ 0, 1\}$, $0<t<T$, and $\mathcal{F}_A , \mathcal{F}_B, \mathcal{F}_S \in [C (\mathbb{R}^3) ]^3$. Assume that for every $\varphi_A , \varphi_B, \varphi_S \in [C^\infty (\mathbb{R}^3) ]^3$ satisfying \eqref{eq28},
\begin{equation*}
\frac{d}{d \varepsilon } \bigg\vert_{\varepsilon =0} E_{D+W} [v_A^\varepsilon, v_B^\varepsilon , v_S^\varepsilon] = \int_{\Omega_A (t)} \mathcal{F}_A \cdot \varphi_A { \ } d x + \int_{\Omega_B (t)} \mathcal{F}_B \cdot \varphi_B { \ } d x + \int_{\Gamma (t)} \mathcal{F}_S \cdot \varphi_S { \ } d \mathcal{H}^2_x.
\end{equation*}
Then
\begin{equation}\label{eq2010}
\begin{cases}
\mathcal{F}_A = {\rm{div}} \mathcal{T}_A (v_A , \pi_A) & \text{ in } \Omega_{A} (t),\\
\mathcal{F}_B = {\rm{div}} \mathcal{T}_B ( v_B , \pi_B) & \text{ in } \Omega_{B} (t),\\
\mathcal{F}_S = {\rm{div}}_\Gamma \mathcal{T}_S (v_S, \pi_S)  + \widetilde{\mathcal{T}}_B (v_B , \pi_B ) n_\Gamma - \widetilde{\mathcal{T}}_A (v_A , \pi_A ) n_\Gamma & \text{ on } \Gamma (t),
\end{cases}
\end{equation}
where $(\mathcal{T}_A, \mathcal{T}_B, \mathcal{T}_S)$ and $( \widetilde{\mathcal{T}}_A, \widetilde{\mathcal{T}}_B)$ are defined by \eqref{eq17} and \eqref{eq18}, respectively.
\end{theorem}

\begin{theorem}[Endothermic energies derived from thermal diffusion]\label{thm27}
Let $0<t<T$, and $Q_A, Q_B,Q_S \in C (\mathbb{R}^3)$. Assume that for every $\psi_A, \psi_B, \psi_S \in C^\infty (\mathbb{R}^3)$ satisfying \eqref{eq29},
\begin{equation*}
\frac{d}{d \varepsilon} \bigg\vert_{\varepsilon = 0} E_{TD} [\theta_A^\varepsilon , \theta_B^\varepsilon, \theta_S^\varepsilon]\\
= \int_{\Omega_A (t)} Q_A \psi_A { \ }d x + \int_{\Omega_B (t)} Q_B \psi_B { \ }d x + \int_{\Gamma (t)} Q_S \psi_S { \ }d \mathcal{H}^2_x. 
\end{equation*}
Then
\begin{equation}\label{eq2011}
\begin{cases}
Q_A = {\rm{div}}( \kappa_A \nabla \theta_A) & \text{ in } \Omega_{A} (t),\\
Q_B = {\rm{div}}( \kappa_B \nabla \theta_B) & \text{ in } \Omega_{B}(t),\\
Q_S = {\rm{div}}_\Gamma ( \kappa_S \nabla_\Gamma \theta_S )  + \kappa_B ( n_\Gamma \cdot \nabla ) \theta_B - \kappa_A ( n_\Gamma \cdot \nabla ) \theta_A& \text{ on } \Gamma (t).
\end{cases}
\end{equation}
\end{theorem}

Finally, we state the conservation laws and thermodynamic potential for our multiphase flow system.
\begin{theorem}[Conservative forms and conservation laws]\label{thm28}Let $r \in \{ 0,1 \}$. Then{ \ }\\
$(\rm{i})$ Any solution to system \eqref{eq12}-\eqref{eq14} satisfies \eqref{eq19}-\eqref{eq1011}.\\
$(\rm{ii})$ Any solution to system \eqref{eq11}-\eqref{eq14} satisfies \eqref{eq1012}-\eqref{eq1014}.\\
$(\rm{iii})$ Assume that $r=1$, and that for $0<t <T$
\begin{equation}\label{eq2022}
\int_{\partial \Omega} \{ \mu_B D (v_B) n_\Omega + \lambda_B ({\rm{div}} v_B) n_\Omega -\pi_B n_\Omega \} { \ } d \mathcal{H}^2_x = { }^t ( 0 , 0 ,0 ).
\end{equation}
Then any solution to \eqref{eq11}-\eqref{eq14} satisfies \eqref{eq1015}.
\end{theorem}
\begin{theorem}[Thermodynamic potential]\label{thm29} Let $\sharp = A , B ,S$, and $r \in \{ 0 , 1 \}$. Suppose that $\rho_\sharp$ and $\theta_\sharp$ are positive functions. Set $h_\sharp =e_\sharp + \pi_\sharp /\rho_\sharp$, $F_\sharp^H = e_\sharp -\theta_\sharp \varsigma_\sharp$, $F_\sharp^G = h_\sharp - \theta_\sharp \varsigma_\sharp$. Assume that $e_\sharp$ satisfies $D_t^\sharp e_\sharp = \theta_\sharp D_t^\sharp \varsigma_\sharp - \pi_\sharp D_t^\sharp (1/\rho_\sharp)$. Then \eqref{eq1016}-\eqref{eq1019} hold.
\end{theorem}

We prove Theorems \ref{thm24}, \ref{thm26}, \ref{thm27} in Section \ref{sect3}, Theorem \ref{thm28} in Section \ref{sect5}, and Theorem \ref{thm29} in Section \ref{sect6}. In Section \ref{sect4} we derive system \eqref{eq12}-\eqref{eq14} by our thermodynamic approaches.

\section{Application of Transport Theorems and Integration by Parts}\label{sect3}

We apply the transport theorems (Definition \ref{def21}) and several formulas for integration by parts (Lemma \ref{lem31}) to prove Theorems \ref{thm24}, \ref{thm26}, and \ref{thm27}.

\begin{lemma}[Formulas for integration by parts]\label{lem31}{ \ }Fix $0 \leq t < T$ and $j=1,2,3$. Then for every $f,g \in C^1 (\mathbb{R}^3)$,
\begin{equation*}
\int_{\Omega_A (t)} (\partial_j f ) g { \ } d x = - \int_{\Omega_A (t)} f (\partial_j g ) { \ } d x + \int_{\Gamma (t)} f  g n_j^\Gamma { \ } d \mathcal{H}^2_x,
\end{equation*}
\begin{equation*}
\int_{\Omega_B (t)} (\partial_j f ) g { \ } d x = - \int_{\Omega_B (t)} f (\partial_j g ) { \ } d x + \int_{\partial \Omega} f  g n_j^\Omega { \ } d \mathcal{H}^2_x - \int_{\Gamma (t)} f  g n_j^\Gamma { \ } d \mathcal{H}^2_x,
\end{equation*}
\begin{equation*}
\int_{\Gamma (t)} (\partial^\Gamma_j f ) g { \ } d \mathcal{H}^2_x = - \int_{\Gamma (t)} f (\partial^\Gamma_j g ) { \ } d \mathcal{H}^2_x - \int_{\Gamma (t)} H_\Gamma f  g n_j^\Gamma { \ } d \mathcal{H}^2_x,
\end{equation*}
where $H_\Gamma = - {\rm{div}}_\Gamma n_\Gamma$ and $\partial^\Gamma_j f = \partial_j f - n_j^\Gamma (n_\Gamma \cdot \nabla )f$. Here $n_\Gamma = n_\Gamma (x,t) = { }^t (n_1^\Gamma , n_2^\Gamma , n_3^\Gamma) $ denotes the unit outer normal vector at $x \in \Gamma (t)$ and $n_\Omega = n_\Omega (x) = { }^t (n_1^\Omega , n_2^\Omega , n_3^\Omega)$ the unit outer normal vector at $x \in \partial \Omega$ (see Figure \ref{Fig1}). 
\end{lemma}
\noindent Applying the Gauss divergence theorem and the surface divergence theorem (Section 9 in \cite{Sim83}, Theorem 2.3 in \cite{K19}), we can prove Lemma \ref{lem31}.

We now make use of the transport theorems to prove Theorem \ref{thm24}.
\begin{proof}[Proof of Theorem \ref{thm24}]
We only prove $(\rm{ii})$ since the proof of $(\rm{i})$ is similar. We assume that $(\rho_A,\rho_B,\rho_S)$ satisfies \eqref{eq12}. Let $\widetilde{Q}_A, \widetilde{Q}_B, \widetilde{Q}_S \in C (\mathbb{R}^4)$ and $p_A,p_B,p_S \in C^1 (\mathbb{R})$. We first show that for every $0 < t <T$ and $\Lambda \subset \Omega$,
\begin{align}
\frac{d}{d t} \int_{\Omega_A (t) \cap \Lambda} \{ \rho_A e_A  - p_A (\rho_A ) \} { \ }d x = \int_{\Omega_A (t) \cap \Lambda} \{ \rho_A D_t^A e_A  + ({\rm{div}} v_A) \Pi_A \} { \ } d x,\label{eq31}\\
\frac{d}{d t} \int_{\Omega_B (t) \cap \Lambda} \{ \rho_B e_B  - p_B (\rho_B ) \} { \ }d x = \int_{\Omega_B (t) \cap \Lambda} \{ \rho_B D_t^B e_B  +  ({\rm{div}} v_B) \Pi_B \} { \ } d x,\label{eq32}\\
\frac{d}{dt} \int_{\Gamma (t) \cap \Lambda} \{ \rho_S e_S  - p_S (\rho_S ) \} { \ }d \mathcal{H}^2_x = \int_{\Gamma (t) \cap \Lambda} \{ \rho_S D_t^S e_S  + ({\rm{div}}_\Gamma v_S) \Pi_S \} { \ } d \mathcal{H}^2_x,\label{eq33}
\end{align}
where $(\Pi_A, \Pi_B, \Pi_S)$ is defined by \eqref{eq26}. We only derive \eqref{eq33}. Using the surface transport theorem \eqref{eq23} and \eqref{eq12}, we check that for $0 < t <T$ and $\Lambda \subset \Omega$
\begin{multline*}
\frac{d}{dt} \int_{\Gamma (t) \cap \Lambda} \{ \rho_S e_S  - p_S (\rho_S ) \} { \ }d \mathcal{H}^2_x = \int_{\Gamma (t) \cap \Lambda} \{ (D_t^S \rho_S ) e_S + \rho_S D_t e_S + \rho_S e_S ({\rm{div}}_\Gamma v_S) \} { \ }d \mathcal{H}^2_x\\
 + \int_{\Gamma (t) \cap \Lambda} \{ - D_t^S \rho_S p'_S (\rho_S) - p_S (\rho_S ) ({\rm{div}}_\Gamma v_S) \}{ \ } d \mathcal{H}^2_x\\
= \int_{\Gamma (t) \cap \Lambda}  \{ \rho_S D_t^S e_S  + (\rho_S p'_S(\rho_S) - p_S(\rho_S) ) ({\rm{div}}_\Gamma v_S) \} { \ } d \mathcal{H}^2_x.
\end{multline*}
Thus, we see \eqref{eq33}. We assume that for every $0 < t <T$ and $\Lambda \subset \Omega$,
\begin{align*}
\frac{d}{d t} \int_{\Omega_A (t) \cap \Lambda} \{ \rho_A e_A  - p_A (\rho_A ) \} { \ }d x = \int_{\Omega_A (t) \cap \Lambda} \widetilde{Q}_A { \ } d x,\\
\frac{d}{d t} \int_{\Omega_B (t) \cap \Lambda} \{ \rho_B e_B  - p_B (\rho_B ) \} { \ }d x = \int_{\Omega_B (t) \cap \Lambda} \widetilde{Q}_B { \ } d x,\\
\frac{d}{dt} \int_{\Gamma (t) \cap \Lambda} \{ \rho_S e_S  - p_S (\rho_S ) \} { \ }d \mathcal{H}^2_x = \int_{\Gamma (t) \cap \Lambda} \widetilde{Q}_S { \ } d \mathcal{H}^2_x.
\end{align*}
Applying \eqref{eq31}-\eqref{eq33}, we have \eqref{eq25}. Therefore, Theorem \ref{thm24} is proved.
\end{proof}

Let us apply integration by parts to prove Theorems \ref{thm26} and \ref{thm27}.
\begin{proof}[Proof of Theorem \ref{thm26}]
Let $r \in \{ 0, 1\}$ and $0<t<T$. Let $\varphi_A , \varphi_B, \varphi_S \in [C^\infty (\mathbb{R}^3) ]^3$ satisfying \eqref{eq28}. A direct calculation gives
\begin{multline*}
\frac{d}{d \varepsilon } \bigg\vert_{\varepsilon =0} E_{D} [v_A^\varepsilon , v_B^\varepsilon , v_S^\varepsilon ] = - \int_{\Omega_A (t)} \{ \mu_A D (v_A): D (\varphi_A) + \lambda_A ({\rm{div}} v_A)({\rm{div}} \varphi_A) \}{ \ } d x\\
- \int_{\Omega_B (t)} \{ \mu_B D (v_B): D (\varphi_B) + \lambda_B ({\rm{div}} v_B)({\rm{div}} \varphi_B) \}{ \ } d x\\ 
- \int_{\Gamma (t)} \{ \mu_S D_\Gamma (v_S): D_\Gamma (\varphi_S) + \lambda_S ({\rm{div}}_\Gamma v_S)({\rm{div}}_\Gamma \varphi_S) \}{ \ } d \mathcal{H}^2_x,
\end{multline*}
and
\begin{multline*}
\frac{d}{d \varepsilon } \bigg\vert_{\varepsilon =0} E_{W} [v_A^\varepsilon , v_B^\varepsilon  , v_S^\varepsilon ]\\ = \int_{\Omega_A (t)} ({\rm{div}} \varphi_A) \pi_A { \ }d x + \int_{\Omega_B (t)} ({\rm{div}} \varphi_B) \pi_B { \ }d x +\int_{\Gamma (t)} ({\rm{div}}_\Gamma \varphi_S) \pi_S { \ }d \mathcal{H}^2_x.
\end{multline*}
Using integration by parts (Lemma \ref{lem31}), we find that
\begin{multline}\label{eq34}
\frac{d}{d \varepsilon } \bigg\vert_{\varepsilon =0} E_{D+W} [v_A^\varepsilon , v_B^\varepsilon , v_S^\varepsilon ] = \int_{\Omega_A (t)}{\rm{div}} \mathcal{T}_A (v_A , \pi_A ) \cdot \varphi_A { \ } d x\\
+ \int_{\Omega_B (t)}{\rm{div}} \mathcal{T}_B (v_B , \pi_B ) \cdot \varphi_B { \ } d x + \int_{\Gamma (t)}{\rm{div}}_\Gamma \mathcal{T}_S (v_S , \pi_S ) \cdot \varphi_S { \ } d \mathcal{H}^2_x\\
- \int_{\Gamma (t)} \{ \mathcal{T}_A (v_A, \pi_A) n_\Gamma \} \cdot \varphi_A { \ } d \mathcal{H}^2_x + \int_{\Gamma (t)} \{ \mathcal{T}_B (v_B , \pi_B) n_\Gamma \} \cdot \varphi_B { \ } d \mathcal{H}^2_x\\
- \int_{\partial \Omega} \{ \mathcal{T}_B ( v_B , \pi_B ) n_\Omega \} \cdot \varphi_B { \ } d \mathcal{H}^2_x,
\end{multline}
where $(\mathcal{T}_A, \mathcal{T}_B, \mathcal{T}_S)$ and $( \widetilde{\mathcal{T}}_A, \widetilde{\mathcal{T}}_B)$ are defined by \eqref{eq17} and \eqref{eq18}. Applying \eqref{eq28}, we check that
\begin{multline*}
\text{(R.H.S) of }\eqref{eq34} = \int_{\Omega_A (t)}{\rm{div}} \mathcal{T}_A (v_A , \pi_A ) \cdot \varphi_A { \ } d x + \int_{\Omega_B (t)}{\rm{div}} \mathcal{T}_B (v_B , \pi_B ) \cdot \varphi_B { \ } d x\\
 + \int_{\Gamma (t)} \{ {\rm{div}}_\Gamma \mathcal{T}_S (v_S , \pi_S ) + \widetilde{\mathcal{T}}_B (v_B , \pi_B ) n_\Gamma - \widetilde{\mathcal{T}}_A (v_A , \pi_A ) n_\Gamma \} \cdot \varphi_S { \ } d \mathcal{H}^2_x.
\end{multline*}
Here we used the facts that
\begin{align*}
D (v_A ) n_\Gamma \cdot \varphi_A = (n_\Gamma \cdot \{ (n_\Gamma \cdot \nabla) v_A \} ) (n_\Gamma \cdot \varphi_S ) \text{ if } r=0,\\
D (v_B ) n_\Gamma \cdot \varphi_B = (n_\Gamma \cdot \{ (n_\Gamma \cdot \nabla) v_B \}) (n_\Gamma \cdot \varphi_S ) \text{ if } r=0.
\end{align*}
From fundamental lemma of calculation of variations, we have \eqref{eq2010}. Therefore, Theorem \ref{thm26} is proved.
\end{proof}

\begin{proof}[Proof of Theorem \ref{thm27}]
Fix $0<t<T$. Let $\psi_A, \psi_B, \psi_S \in C^\infty (\mathbb{R}^3)$ satisfying \eqref{eq29}. Since
\begin{multline*}
\frac{d}{d \varepsilon } \bigg\vert_{\varepsilon =0} E_{TD} [\theta_A^\varepsilon , \theta_B^\varepsilon , \theta_S^\varepsilon ]\\
 = - \int_{\Omega_A (t)} \kappa_A \nabla \theta_A \cdot \nabla \psi_A { \ }d x - \int_{\Omega_B (t)} \kappa_B \nabla \theta_B \cdot \nabla \psi_B { \ }d x - \int_{\Gamma (t)} \kappa_S \nabla_\Gamma \theta_S \cdot \nabla_\Gamma \psi_S { \ }d \mathcal{H}^2_x,
\end{multline*}
we apply the integration by parts (Lemma \ref{lem31}), \eqref{eq29}, and \eqref{eq11} to see that
\begin{multline*}
\frac{d}{d \varepsilon } \bigg\vert_{\varepsilon =0} E_{TD} [\theta_A^\varepsilon , \theta_B^\varepsilon , \theta_S^\varepsilon ] = \int_{\Omega_A (t)} {\rm{div}} (\kappa_A \nabla \theta_A) \psi_A { \ }d x + \int_{\Omega_B (t)} {\rm{div}} (\kappa_B \nabla \theta_B)\psi_B { \ }d x\\
 + \int_{\Gamma (t)} \{ {\rm{div}}_\Gamma ( \kappa_S \nabla_\Gamma \theta_S ) + \kappa_B (n_\Gamma \cdot \nabla ) \theta_B - \kappa_A (n_\Gamma \cdot \nabla ) \theta_A  \} \psi_S { \ }d \mathcal{H}^2_x.
\end{multline*}
From fundamental lemma of calculation of variations, we have \eqref{eq2011}. Therefore, Theorem \ref{thm27} is proved.
\end{proof}

\section{Thermodynamical Modeling}\label{sect4}
In this section we make mathematical models for multiphase flow with surface tension and flow by our thermodynamic approaches. Under the restrictions \eqref{eq11}, we apply Proposition \ref{prop23}, the first law of thermodynamics (Theorem \ref{thm24}), and our energy densities (Definition \ref{def22}) to derive equations \eqref{eq12}-\eqref{eq14}. In this section we consider the case when the fluids in $\Omega_{A,T}$, $\Omega_{B,T}$, $\Gamma_T$ are barotropic fluids.

Let $r \in \{ 0, 1 \}$, and $p_A , p_B, p_B \in C^1 (\mathbb{R})$. We assume that $(v_A , v_B , v_S, \theta_A , \theta_B , \theta_S )$ satisfies \eqref{eq11}. We consider the energy densities defined in Definition \ref{def22} as the energy densities for multiphase flow with surface flow.

From Proposition \ref{prop23}, we admit that system \eqref{eq12} is the continuity equations of our system, that is, we assume that $(\rho_A , \rho_B, \rho_S)$ satisfies \eqref{eq12}. 

Let $(Q_A, Q_B, Q_S)$ be the endothermic energies derived from energies dissipation due to thermal diffusion. From Theorem \ref{thm27}, we set
\begin{equation*}
\begin{cases}
Q_A = {\rm{div}}( \kappa_A \nabla \theta_A) & \text{ in } \Omega_{A,T},\\
Q_B = {\rm{div}}( \kappa_B \nabla \theta_B) & \text{ in } \Omega_{B,T},\\
Q_S = {\rm{div}}_\Gamma ( \kappa_S \nabla_\Gamma \theta_S )  + \kappa_B ( n_\Gamma \cdot \nabla ) \theta_B - \kappa_A ( n_\Gamma \cdot \nabla ) \theta_A & \text{ on } \Gamma_T.
\end{cases}
\end{equation*}
Let $\widetilde{Q}_\sharp = \widetilde{Q}_\sharp (x,t)$ be the quantity of heat supplied the fluid in $\Omega_\sharp (t)$, where $\sharp = A, B, S$, and $\Omega_S (t) := \Gamma (t)$. Since $e_{D_\sharp}$ is the energy density for energy dissipation due to the viscosities $(\mu_\sharp , \lambda_\sharp)$, we set $\widetilde{Q}_\sharp = Q_\sharp + e_{D_\sharp}$. Now we admit the first law of thermodynamics, that is, suppose that for every $0 < t <T$ and $\Lambda \subset \Omega$,
\begin{align*}
\frac{d}{d t} \int_{\Omega_A (t) \cap \Lambda} \{ \rho_A e_A - p_A (\rho_A) \} { \ }d x = \int_{\Omega_A (t) \cap \Lambda} (Q_A + e_{D_A}) { \ } d x,\\
\frac{d}{d t} \int_{\Omega_B (t) \cap \Lambda} \{ \rho_B e_B - p_B (\rho_B) \}  { \ }d x = \int_{\Omega_B (t) \cap \Lambda} (Q_B + e_{D_B}  ){ \ } d x,\\
\frac{d}{dt} \int_{\Gamma (t) \cap \Lambda} \{ \rho_S e_S - p_S (\rho_S) \} { \ }d \mathcal{H}^2_x = \int_{\Gamma (t) \cap \Lambda} ( Q_S + e_{D_S} ) { \ } d \mathcal{H}^2_x.
\end{align*}
From Theorem \ref{thm24}, we have
\begin{equation}\label{eq41}
\begin{cases}
\rho_A D_t^A e_A + ({\rm{div}} v_A) \Pi_A = {\rm{div}} q_A + e_{D_A}  \text{ in }\Omega_{A,T} ,\\
\rho_B D_t^B e_B + ({\rm{div}} v_B ) \Pi_B  = {\rm{div}} q_B +e_{D_B}  \text{ in }\Omega_{B,T},\\
\rho_S D_t^S e_S + ({\rm{div}}_\Gamma v_S) \Pi_S = {\rm{div}}_\Gamma q_S + e_{D_S} +  q_B \cdot n_\Gamma - q_A \cdot n_\Gamma \text{ on }\Gamma_T,
\end{cases}
\end{equation}
where $(q_A,q_B , q_S )$ and $(\Pi_A , \Pi_B , \Pi_S)$ are defined by \eqref{eq15} and \eqref{eq26}.

Let $\mathcal{F}^A, \mathcal{F}^B, \mathcal{F}^S \in [ C ( \mathbb{R}^4) ]^3$. We assume that the momentum equations of our system are written by
\begin{equation}\label{eq42}
\begin{cases}
\rho_A D_t^A v_A = \mathcal{F}^A  & \text{ in }\Omega_{A,T},\\
\rho_B D_t^B v_B = \mathcal{F}^B & \text{ in }\Omega_{B , T},\\
\rho_S D_t^S v_S = \mathcal{F}^S & \text{ on }\Gamma_T.
\end{cases}
\end{equation}
From a thermodynamic point of view we assume that our system satisfies the conservation law of total energy, that is, $(\mathcal{F}^A , \mathcal{F}^B , \mathcal{F}^S)$ satisfies that for each $0<t <T$, 
\begin{multline}\label{eq43}
\frac{d}{d t} \bigg\{ \int_{\Omega_A (t)} \bigg( \frac{1}{2} \rho_A \vert v_A \vert^2 + \rho_A e_A \bigg) { \ } dx +  \int_{\Omega_B (t)} \bigg( \frac{1}{2} \rho_B \vert v_B \vert^2 + \rho_B e_B \bigg) { \ } dx\\
 + \int_{\Gamma (t)} \bigg( \frac{1}{2} \rho_S \vert v_S \vert^2 + \rho_S e_S \bigg) { \ } d \mathcal{H}^2_x \bigg\} = 0.
\end{multline}
Using the transport theorems with \eqref{eq12}, \eqref{eq41}, \eqref{eq42}, we see that
\begin{multline*}
\text{(L.H.S) of } \eqref{eq43} = \int_{\Omega_A (t)} ( \rho_A D_t^A v_A \cdot v_A + \rho_A D_t^A e_A ) { \ } dx\\
 +  \int_{\Omega_B (t)} ( \rho_B D_t^B v_B \cdot v_B + \rho_B D_t^B e_B ) { \ } dx + \int_{\Gamma (t)} ( \rho_S D_t^S v_S \cdot v_S + \rho_S D_t^S e_S ) { \ } d \mathcal{H}^2_x\\
= \int_{\Omega_A (t)} \{ \mathcal{F}^A \cdot v_A + {\rm{div}} q_A + e_{D_A} - ({\rm{div}} v_A) \Pi_A \} { \ } dx\\
+ \int_{\Omega_B (t)} \{ \mathcal{F}^B \cdot v_B + {\rm{div}} q_B + e_{D_B} - ({\rm{div}} v_B) \Pi_B \} { \ } dx\\
+ \int_{\Gamma (t)} \{ \mathcal{F}^S \cdot v_S + {\rm{div}}_\Gamma q_S + e_{D_S} - ({\rm{div}}_\Gamma v_S) \Pi_S + q_B \cdot n_\Gamma - q_A \cdot n_\Gamma \} { \ } d \mathcal{H}^2_x.
\end{multline*}
Applying the integration by parts with \eqref{eq11}, we observe that
\begin{multline*}
\text{(L.H.S) of } \eqref{eq43} = \int_{\Omega_A (t)} \{ \mathcal{F}^A - {\rm{div}} \mathcal{T}_A (v_A , \Pi_A ) \} \cdot v_A { \ } d x\\
 + \int_{\Omega_B (t)} \{ \mathcal{F}^B - {\rm{div}} \mathcal{T}_B (v_B , \Pi_B ) \} \cdot v_B { \ } d x + \int_{\Gamma (t)} \{ \mathcal{F}^S - {\rm{div}}_\Gamma \mathcal{T}_S (v_S , \Pi_S ) \} \cdot v_S { \ } d \mathcal{H}^2_x\\
+ \int_{\Gamma (t)} \widetilde{\mathcal{T}}_A (v_A, \Pi_A ) n_\Gamma \cdot v_S { \ } d \mathcal{H}^2_x - \int_{\Gamma (t)} \widetilde{\mathcal{T}}_B (v_B, \Pi_B ) n_\Gamma \cdot v_S { \ } d \mathcal{H}^2_x,
\end{multline*}
where $(\mathcal{T}_A, \mathcal{T}_B, \mathcal{T}_S)$ and $( \widetilde{\mathcal{T}}_A, \widetilde{\mathcal{T}}_B)$ are defined by \eqref{eq17} and \eqref{eq18}. Here we used the facts that
\begin{align*}
\int_{\Gamma (t)} \{ \mathcal{T}_A (v_A , \Pi_A) n_\Gamma \} \cdot v_A { \ } d \mathcal{H}^2_x = \int_{\Gamma (t)} \widetilde{\mathcal{T}}_A (v_A, \Pi_A ) n_\Gamma \cdot v_S { \ } d \mathcal{H}^2_x,\\
\int_{\Gamma (t)} \{ \mathcal{T}_B (v_B , \Pi_B) n_\Gamma \} \cdot v_B { \ } d \mathcal{H}^2_x = \int_{\Gamma (t)} \widetilde{\mathcal{T}}_B (v_B, \Pi_B ) n_\Gamma \cdot v_S { \ } d \mathcal{H}^2_x.
\end{align*}
Thus, we set
\begin{equation}\label{eq44}
\begin{cases}
\mathcal{F}^A = {\rm{div}} \mathcal{T}_A (v_A , \Pi_A ),\\
\mathcal{F}^B = {\rm{div}} \mathcal{T}_B (v_B , \Pi_B),\\
\mathcal{F}^S = {\rm{div}}_\Gamma \mathcal{T}_S (v_S , \Pi_S ) + \widetilde{\mathcal{T}}_B (v_B , \Pi_B) n_\Gamma - \widetilde{\mathcal{T}}_A (v_A, \Pi_A) n_\Gamma
\end{cases}
\end{equation}
to see that (L.H.S) of \eqref{eq43} equals to zero. Combining \eqref{eq41}, \eqref{eq42}, \eqref{eq44}, we have \eqref{eq13} and \eqref{eq14}. Therefore, we derive \eqref{eq12}-\eqref{eq14} by our thermodynamic approach.

Finally, we introduce another approach to derive the momentum equations \eqref{eq14}. We assume that the time rate of change of the momentum equals to the forces derived from the variation of energies dissipation due to the viscosities, that is, suppose that for every $0 < t <T$ and $\Lambda \subset \Omega$,
\begin{align*}
\frac{d}{d t} \int_{\Omega_A (t) \cap \Lambda} \rho_A v_A { \ }d x = \int_{\Omega_A (t) \cap \Lambda} \mathcal{F}_A { \ } d x,\\
\frac{d}{d t} \int_{\Omega_B (t) \cap \Lambda} \rho_B v_B  { \ }d x = \int_{\Omega_B (t) \cap \Lambda} \mathcal{F}_B { \ } d x,\\
\frac{d}{dt} \int_{\Gamma (t) \cap \Lambda} \rho_S v_S { \ }d \mathcal{H}^2_x = \int_{\Gamma (t) \cap \Lambda} \mathcal{F}_S { \ } d \mathcal{H}^2_x,
\end{align*}
where $(\mathcal{F}_A, \mathcal{F}_B, \mathcal{F}_S)$ is defined by \eqref{eq2010}. Using the transport theorems with \eqref{eq12}, we have \eqref{eq14}.

\section{Conservative Forms and Conservation Laws}\label{sect5}
We study the conservation laws of our model to prove Theorem \ref{thm28}.
\begin{proof}[Proof of Theorem \ref{thm28}]
Let $r \in \{0, 1\}$. Direct calculations give $(\rm{i})$ (see Remark \ref{rem13}). We now prove $(\rm{ii})$. From Proposition \ref{prop23} and the arguments in Section \ref{sect4}, we see \eqref{eq1012} and \eqref{eq1013}. Using the transport theorems (Definition \ref{def21}) with \eqref{eq12} and \eqref{eq14}, we see that
\begin{multline}\label{eq51}
\frac{d}{d t} \bigg( \int_{\Omega_A (t)} \frac{1}{2} \rho_A \vert v_A \vert^2 { \ }d x + \int_{\Omega_B (t)} \frac{1}{2} \rho_B \vert v_B \vert^2 { \ }d x + \int_{\Gamma (t)} \frac{1}{2} \rho_S \vert v_S \vert^2 { \ }d \mathcal{H}^2_x \bigg)\\
= \int_{\Omega_A (t)} \rho_A D_t^A v_A \cdot v_A { \ }d x + \int_{\Omega_B (t)} \rho_B D_t^B v_B \cdot v_B { \ }d x + \int_{\Gamma (t)} \rho_S D_t^S v_S \cdot v_S { \ }d \mathcal{H}^2_x\\
= \int_{\Omega_A (t)} {\rm{div}} \mathcal{T}_A (v_A , \pi_A ) \cdot v_A { \ }d x + \int_{\Omega_B (t)} {\rm{div}}\mathcal{T}_B (v_B , \pi_B ) \cdot v_B { \ }d x\\
 + \int_{\Gamma (t)}\{ {\rm{div}}_\Gamma \mathcal{T}_S (v_S , \pi_S)  + \widetilde{\mathcal{T}}_B(v_B , \pi_B) n_\Gamma - \widetilde{\mathcal{T}}_A (v_A , \pi_A ) n_\Gamma \} \cdot v_S { \ }d \mathcal{H}^2_x.
\end{multline}
Applying integration by parts (Lemma \ref{lem31}) and \eqref{eq11}, we check that
\begin{multline*}
\text{(R.H.S) of } \eqref{eq51} = \int_{\Omega_A (t)} \{ - e_{D_A} + ({\rm{div}} v_A ) \pi_A \} { \ }d x\\
 + \int_{\Omega_B (t)} \{ - e_{D_B} + ({\rm{div}} v_B ) \pi_B \} { \ }d x + \int_{\Gamma (t)} \{ - e_{D_S} + ({\rm{div}}_\Gamma v_S )\pi_S \} { \ }d \mathcal{H}^2_x,
\end{multline*}
where $(e_{D_A}, e_{D_B}, e_{D_S})$ is defined by \eqref{eq16}. Integrating with respect to $t$, we have \eqref{eq1014}.

Finally, we show $(\rm{iii})$. Assume that $r =1$. Using the transport and divergence theorems (Definition \ref{def21} and Lemma \ref{lem31}) with \eqref{eq2022}, we see that
\begin{multline*}
\frac{d}{d t} \bigg( \int_{\Omega_A (t)} \rho_A v_A { \ }d x + \int_{\Omega_B (t)} \rho_B v_B { \ }d x + \int_{\Gamma (t)} \rho_S v_S { \ }d \mathcal{H}^2_x \bigg)\\
= \int_{\Omega_A (t)} {\rm{div}} \mathcal{T}_A (v_A , \pi_A ) { \ }d x + \int_{\Omega_B (t)} {\rm{div}}\mathcal{T}_B (v_B , \pi_B ) { \ }d x\\
 + \int_{\Gamma (t)}\{ {\rm{div}}_\Gamma \mathcal{T}_S (v_S , \pi_S )  + \mathcal{T}_B (v_B , \pi_B) n_\Gamma - \mathcal{T}_A (v_A , \pi_A) n_\Gamma \} { \ }d \mathcal{H}^2_x = { }^t (0,0,0).
\end{multline*}
Integrating with respect to $t$, we have \eqref{eq1015}. Therefore, Theorem \ref{thm28} is proved.
\end{proof}

\section{Thermodynamic Potential}\label{sect6}
We investigate thermodynamic potential for our model to prove Theorem \ref{thm29}. 
\begin{proof}[Proof of Theorem \ref{thm29}]
We only prove the case when $\sharp = S$. Let $r \in \{0 , 1 \}$. Assume that $\rho_S$ and $\theta_S$ are positive functions. Set $h_S =e_S + \pi_S /\rho_S$, $F_S^H = e_S -\theta_S \varsigma_S$, $F_S^G = h_S - \theta_S \varsigma_S$. Assume that $e_S$ satisfies the thermodynamic identity:
\begin{equation}\label{eq61}
D_t^S e_S = \theta_S D_t^S \varsigma_S - \pi_S D_t^S \bigg( \frac{1}{\rho_S} \bigg).
\end{equation}

We first derive \eqref{eq1016}. By \eqref{eq12} and \eqref{eq13}, we see that
\begin{align*}
\rho_S D_t^S h_S & = \rho_S D_t^S e_S + \rho_S D_t^S \bigg( \frac{\pi_S}{\rho_S} \bigg)\\
& = {\rm{div}}_\Gamma q_S + e_{D_S} + q_B \cdot n_\Gamma - q_A \cdot n_\Gamma + D_t^S \pi_S.
\end{align*}
This shows that
\begin{align*}
D_t^N (\rho_S h_S) + {\rm{div}}_\Gamma ( \rho_S h_S v_S - q_S ) & = \rho_S D_t^S h_S - {\rm{div}}_\Gamma q_S\\
& = e_{D_S} + D_t^S \pi_S + q_B \cdot n_\Gamma - q_A \cdot n_\Gamma ,
\end{align*}
which is \eqref{eq1016}.

Next we show \eqref{eq1017}. From \eqref{eq13} and \eqref{eq61}, we find that
\begin{align*}
\theta_S \rho_S D_t^S \varsigma_S  & = \rho_S D_t^S e_S + \rho_S \pi_S D_t^S \bigg( \frac{1}{\rho_S} \bigg)\\& =  {\rm{div}}_\Gamma q_S + e_{D_S} + q_B \cdot n_\Gamma - q_A \cdot n_\Gamma.
\end{align*}
Using the above equality and \eqref{eq1016}, we check that
\begin{align*}
D_t^N (\rho_S \varsigma_S) + {\rm{div}}_\Gamma \bigg( \rho_S \varsigma_S v_S - \frac{q_S}{\theta_S} \bigg) & = \rho_S D_t^S h_S - {\rm{div}}_\Gamma \bigg( \frac{q_S}{\theta_S} \bigg)\\
& = \frac{e_{D_S}}{\theta_S} + \frac{q_S \cdot {\rm{grad}}_\Gamma \theta_S }{\theta_S^2} + \frac{ q_B \cdot n_\Gamma - q_A \cdot n_\Gamma}{\theta_S}.
\end{align*}
Thus, we have \eqref{eq1017}.

Finally, we derive \eqref{eq1018} and \eqref{eq1019}. Applying \eqref{eq61} and \eqref{eq12}, we see that
\begin{align*}
\rho_S D_t^S F_S^H + \rho_S \varsigma_S D_t^S \theta_S & = \rho_S D_t^S e_S - \rho_S \theta_S D_t^S \varsigma_S\\
& = - ({\rm{div}}_\Gamma v_S) \pi_S,
\end{align*}
and that
\begin{align*}
\rho_S D_t^S F_S^G + \rho_S \varsigma_S D_t^S \theta_S & = \rho_S D_t^S h_S - \rho_S \theta_S D_t^S \varsigma_S\\
& = D_t^S \pi_S.
\end{align*}
Therefore, Theorem \ref{thm29} is proved.
\end{proof}

\begin{flushleft}
{\bf{Data Availability :}} The author declares that data sharing not applicable to this article as no datasets were generated or analyzed during the current study.
\end{flushleft}

\begin{flushleft}
{\bf{Conflict of interest :}} The author declares no conflict of interest associated with this manuscript.
\end{flushleft}

\begin{flushleft}
{\bf{Acknowledgments :}} This work was partly supported by the Japan Society for the Promotion of Science (JSPS) KAKENHI Grant Number JP21K03326.
\end{flushleft}

\end{document}